\documentclass[journal]{article}

\usepackage{graphicx,url}
\usepackage{dcolumn}
\usepackage{bm}
\usepackage{amsmath,amsfonts,amssymb,citesort}

\newtheorem{definition}{\noindent{\it Definition}}[section]
\newtheorem{theorem}{\noindent{\it Theorem}}[section]

\newtheorem{remark}[theorem]{\noindent{\it Remark}}
\newtheorem{corollary}[theorem]{\noindent{\it Corollary}}
\newenvironment{proof}{\noindent{\it Proof:}}{$\hfill$ $\Box$\\ }
\newtheorem{example}{\noindent{\it Example}}[section]

\begin{document}

\title{Convolutional Codes Derived From Group Character Codes}

 \author{Giuliano G. La Guardia
\thanks{Giuliano Gadioli La Guardia is with Department of Mathematics and Statistics,
State University of Ponta Grossa (UEPG), 84030-900, Ponta Grossa,
PR, Brazil. }}

\maketitle

\begin{abstract}
New families of unit memory as well as multi-memory convolutional
codes are constructed algebraically in this paper. These
convolutional codes are derived from the class of group character
codes. The proposed codes have basic generator matrices,
consequently, they are non catastrophic. Additionally, the new
code parameters are better than the ones available in the
literature.
\end{abstract}

\section{Introduction}\label{Intro}

Constructions of (classical) convolutional codes and their
corresponding properties have been presented in the literature
\cite{Forney:1970,Piret:1983,Piret:1988,Piret:1988art,Rosenthal:1999,York:1999,Hole:2000,Rosenthal:2001,Gluesing:2006,Luerssen:2006,Schmale:2006,Aly:2007,Luerssen:2008,LaGuardia:2012}.
In \cite{Forney:1970}, the author constructed an algebraic
structure for convolutional codes. Addressing the construction of
maximum-distance-separable (MDS) convolutional codes (in the sense
that the codes attain the generalized Singleton bound introduced
in \cite[Theorem 2.2]{Rosenthal:1999}), there exist interesting
papers in the literature
\cite{Rosenthal:1999,Rosenthal:2001,Schmale:2006}. Concerning the
optimality with respect to other bounds we have
\cite{Piret:1983,Piret:1988art}, and in \cite{Gluesing:2006},
Strongly MDS convolutional codes were constructed. In
\cite{York:1999,Hole:2000,Aly:2007,LaGuardia:2012}, the authors
presented constructions of convolutional BCH codes. In
\cite{Luerssen:2006}, doubly-cyclic convolutional codes were
constructed and in \cite{Luerssen:2008}, the authors described
cyclic convolutional codes by means of the matrix ring.

In this paper we construct families of unit memory as well as
multi-memory convolutional codes, although it is well known that
unit memory codes have large free distance when compared to
multi-memory codes of same rate ( see \cite{Lee:1976}). Our
constructions are performed algebraically and not by computation
search. Consequently, we do not restrict ourselves in constructing
only few specific codes. To do so we apply the famous method
proposed by Piret \cite{Piret:1988} and recently generalized by
Aly \emph{et al.} \cite[Theorem 3]{Aly:2007}, which consists in
the construction of convolutional codes derived from block codes.
The block codes utilized in our construction is the subclass of
$2$-group character codes introduced by Ding \emph{et al.}
\cite{Ding:2000} as well as its generalization to $n$-group,
($n\geq 2$) character codes \cite{Ling:2004}.

The new families of convolutional codes consist of codes whose
parameters are given by\\
\begin{itemize}

\item $(2^{m}, 2^{m}- s_{m}(u), s_{m}(u)- s_{m}(r) ; 1, d_{f}\geq
2^{r+1} {)}_{q}$,
\end{itemize}

\vspace{0.25cm}

\begin{itemize}
\item $(2^{m}, s_{m}(u), s_{m}(u)- s_{m}(r) ; \mu,
d_{f}^{\perp}\geq 2^{m-u}+1 {)}_{q}$,
\end{itemize}

\vspace{0.25cm}

where $q$ is a power of an odd prime, $m\geq 3$ is an integer, $r$
and $u$ are positive integers satisfying $r < u < m$ and
$\displaystyle\sum_{i=u+1}^{m}\left(
\begin{array}{c}
m\\
i\\
\end{array}
\right) > \displaystyle\sum_{i=r+1}^{u}\left(
\begin{array}{c}
m\\
i\\
\end{array}
\right)$, $\mu \geq 1$ is an integer and $s_{m}(v)=
\displaystyle\sum_{i=0}^{v}\left(
\begin{array}{c}
m\\
i\\
\end{array}
\right)$;

\vspace{0.25cm}

\begin{itemize}
\item $(2^{m}, 2^{m}- s_{m}(u), \delta ; 2, d_{f}\geq 2^{r+1}
{)}_{q}$,
\end{itemize}

\vspace{0.25cm}

where $q$ is a power of an odd prime, $m\geq 4$ is
an integer, $s_{m}(u)$ is given above, $\delta =
\displaystyle\sum_{i=r+1}^{v}\left(
\begin{array}{c}
m\\
i\\
\end{array}
\right)$, and $r, u, v$ are positive integers such that the
inequalities $r < v < u < m$, $\displaystyle\sum_{i=u+1}^{m}\left(
\begin{array}{c}
m\\
i\\
\end{array}
\right) \geq \displaystyle\sum_{i=r+1}^{v}\left(
\begin{array}{c}
m\\
i\\
\end{array}
\right) \geq \displaystyle\sum_{i=v+1}^{u}\left(
\begin{array}{c}
m\\
i\\
\end{array}
\right)$ hold;

\vspace{0.35cm}

\begin{itemize}
\item $(l^{m}, l^{m}-S_{m}(u), S_{m}(u)- S_{m}(r); 1, d_{f}\geq
(b+2)l^{a} {)}_{q}$,
\end{itemize}

\vspace{0.35cm}
where $m\geq 3$, $l\geq 3$ are integers, $q$ is a
prime power such that $l | (q - 1)$, $r$ and $u$ are positive
integers satisfying the inequalities $r < u < m(l-1)$ and
$\displaystyle\sum_{i=u+1}^{m}\left(
\begin{array}{c}
m\\
i\\
\end{array}
\right)_{l} \geq \displaystyle\sum_{i=r+1}^{u}\left(
\begin{array}{c}
m\\
i\\
\end{array}
\right)_{l}$, $a$ and $b$ are integers such that $r = a(l- 1) +
b$, $0 \leq b \leq l-2$ and $S_{m}(v)= \displaystyle\sum_{i=0}^{v}
\displaystyle\sum_{k=0}^{m}{(-1)}^{k}\left(
\begin{array}{c}
m\\
k\\
\end{array}
\right)\left(
\begin{array}{c}
m-1+i-kl\\
m-1\\
\end{array}
\right)$.

\vspace{0.25cm}

The paper is organized as follows. In Section~\ref{II}, we recall
basic concepts and results concerning the class of character
codes. Section~\ref{III} deals with basic definitions and known
results on convolutional codes. In Section~\ref{IV}, the
contributions of the work are presented, that is, the
constructions of new families of convolutional codes derived from
character codes. In Section~\ref{V}, we compare the new code
parameters with the ones available in the literature, and finally,
in Section~\ref{VI}, a summary of the paper is described.

\section{Character Codes}\label{II}

Throughout this paper, $p$ denotes a prime number, $q$ is a prime
power and ${\mathbb F}_{q}$ is a finite field with $q$ elements.
As usual, the parameters of a linear code are given by ${[n, k,
d]}_{q}$, and the notation wt$_{H} (x)$ means the Hamming weight
of a vector ${\bf x} \in {\mathbb F}_{q}^{n}$.

The class of group character codes were introduced by Ding
\emph{et al.} \cite{Ding:2000}. These codes are defined by using
characters of groups; they are linear, defined over ${\mathbb
F}_{q}$ and are similar (with respect to the parameters) to binary
Reed-Muller codes.

In order to define such class of codes, let us consider an abelian
group $(G, +)$ of order $n$ and exponent $N$ and let ${\mathbb
F}_{q}$ be a finite field such that $\gcd (n, q)=1$ and $N|(q-1)$.
Assume also that $({\mathbb F}_{q}^{*}, \cdot )$ is the
multiplicative group of nonzero elements of ${\mathbb F}_{q}$.
Then a character ${\gamma}$ from $(G, +)$ to $({\mathbb
F}_{q}^{*}, \cdot )$ is a homomorphism of groups (in which, from
our assumptions, the operation among characters is the
multiplication). We denote the group of characters by $(\Gamma,
\cdot)$. Since in this case $(G, +)$ is isomorphic to $(\Gamma,
\cdot)$, then there exists a bijection $g\in G\longrightarrow
{\gamma}_{g}\in \Gamma$ and we can denote $\Gamma=\{ {\gamma}_{0},
\ldots , {\gamma}_{n-1}\}$ (${\gamma}_{0}$ is the trivial
character). For every $X\subset G$, the \emph{group character
code} $C_{X}$ is a linear code over ${\mathbb F}_{q}$ defined by
$C_{X}=\left\{(c_0, \ldots, c_{n-1})\in {\mathbb F}_{q}^{n}|
\displaystyle\sum_{i=0}^{n-1} c_{i}{\gamma}_{i}(x)=0, \forall x
\in X\right\}$ and has parameters ${[n, k]}_{q}$, where $n=|G|$
and $k=n-|X|$ (see \cite{Ding:2000}). If $X=\{ x_0 , \ldots ,
x_{s-1} \}$ $\subset G$ then a generator matrix for $C_{X}$ is
given by $G_{X}=[{\gamma}_{j-1}(-x_{s-1+i})]_{1\leq i\leq n-s,
1\leq j\leq n};$ a parity check matrix for $C_{X}$ is given by
$H_{X}=[{\gamma}_{j-1}(x_{i-1})]_{1\leq i\leq s, 1\leq j\leq n}$.

A particular case is when it is considered the commutative group
$({\mathbb{Z}}_{2}^{m}, +)$, with $m\geq 1$, and a finite field
${\mathbb F}_{q}$ of odd characteristic. The characters of
${\mathbb{Z}}_{2}^{m}$ are given by ${\gamma}_{{\bf x}}({\bf
y})={(-1)}^{{\bf x}\cdot {\bf y}}$, where ${\bf x}, {\bf y} \in
{\mathbb{Z}}_{2}^{m}$. Then one defines the character code
$C_{q}(r, m)=C_{X_{r}}$, where $X_{r}= \{ {\bf x} \in
{\mathbb{Z}}_{2}^{m} | wt_{H} ({\bf x}) > r \}$, with parameters
$[2^{m},$ $s_{m}(r), 2^{m-r}{]}_{q}$ (see \cite[Theorem
6]{Ding:2000}), where $s_{m}(r)= \displaystyle\sum_{i=0}^{r}\left(
\begin{array}{c}
m\\
i\\
\end{array}
\right)$. Its (Euclidean) dual code ${[C_{q}(r, m)]}^{\perp}$ is
equivalent to $C_{q}(m-r-1, m)$ (see \cite[Theorem 8]{Ding:2000}).

In 2004, Ling \cite{Ling:2004} generalized such class of codes by
considering the group $({\mathbb{Z}}_{l}^{m}, +)$, where $m\geq 1$
and $l\geq 2$ are integers and $F_{q}$ is a finite field that
contains a $l$th root of unity, that is, $l | (q - 1)$. Let ${\bf
x} = (x_1, \ldots, x_{m}) \in {\mathbb{Z}}_{l}^{m}$ and assume
that $||x||= x_1 + \ldots + x_{m}$, where the sum is considered as
a rational integer. Analogously to $X_{r}$, one can define the set
$X(r, m; l) = \{{\bf x} \in {\mathbb{Z}}_{l}^{m} : ||{\bf x}|| > r
\}$, generating the linear $q$-ary group character code $C_{q} (r,
m; l) =\left\{(c_0, \ldots, c_{l^{m}-1})\in {\mathbb
F}_{q}^{l^{m}}| \displaystyle\sum_{i=0}^{l^{m}-1}
c_{i}{\gamma}_{i}({\bf x})=0, \forall {\bf x} \in X(r, m;
l)\right\},$ where ${\gamma}_{0}, \ldots, {\gamma}_{l^{m}-1}$, are
all the characters from ${\mathbb{Z}}_{l}^{m}$ to ${\mathbb
F}_{q}^{*}$. To be more precise, if $\xi$ is a fixed $l$th root of
unity, then the characters ${\gamma}_{i}: {\mathbb{Z}}_{l}^{m}
\longrightarrow {\mathbb F}_{q}^{*}$, $i=0, \ldots, l^{m}-1$, are
given by ${\gamma}_{i}((x_1 , \ldots , x_{m})) = {\xi}^{x_1 i_1 +
\ldots + x_{m}i_{m}}$, where the coefficients $i_{k}$, $k=1,
\ldots , m$, are the coefficients of the (unique) $l$-adic
expansion of $i$. The code $C_{q} (r, m; l)$ has parameters ${[
l^{m}, S_{m}(r), (l - b)l^{m-1-a}]}_{q}$, where $0\leq r < m(l-1)$
is writing as $r = a(l- 1) + b$, $0 \leq b \leq l-2$, and
$S_{m}(r)= \displaystyle\sum_{i=0}^{r}
\displaystyle\sum_{k=0}^{m}{(-1)}^{k}\left(
\begin{array}{c}
m\\
k\\
\end{array}
\right)\left(
\begin{array}{c}
m-1+i-kl\\
m-1\\
\end{array}
\right).$ Furthermore, its (Euclidean) dual code ${[C_{q} (r, m;
l)]}^{\perp}$ is monomial equivalent to $C_{q} (m(l-1)-1-r, m; l)$
(see \cite{Ling:2004}).

\section{Convolutional Codes}\label{III}

In this section we present a brief review of convolutional codes.
For more details we refer the reader to
\cite{Piret:1988,Johannesson:1999,Huffman:2003}.

Recall that a polynomial encoder matrix $G(D) \in {\mathbb
F}_{q}{[D]}^{k \times n}$ is called \emph{basic} if $G(D)$ has a
polynomial right inverse. A basic generator matrix of a
convolutional code $C$ is called \emph{reduced} (or minimal
\cite{Rosenthal:2001,Huffman:2003,Luerssen:2008}) if the overall
constraint length $\delta =\displaystyle\sum_{i=1}^{k} {\delta}_i$
has the smallest value among all basic generator matrices of $C$;
in this case the overall constraint length $\delta$ is called the
\emph{degree} of the code. The \emph{weight} of an element ${\bf
v}(D)\in {\mathbb F}_{q} {[D]}^{n}$ is defined as wt$({\bf
v}(D))=\displaystyle\sum_{i=1}^{n}$wt$(v_i(D))$, where
wt$(v_i(D))$ is the number of nonzero coefficients of $v_{i}(D)$.

\begin{definition}\cite{Johannesson:1999}
A rate $k/n$ \emph{convolutional code} $C$ with parameters $(n, k,
\delta ; \mu,$ $d_{f} {)}_{q}$ is a submodule of ${\mathbb F}_q
{[D]}^{n}$ generated by a reduced basic matrix $G(D)=(g_{ij}) \in
{\mathbb F}_q {[D]}^{k \times n}$, that is, $C = \{ {\bf u}(D)G(D)
| {\bf u}(D)\in {\mathbb F}_{q} {[D]}^{k} \}$, where $n$ is the
length, $k$ is the dimension, $\delta =\displaystyle\sum_{i=1}^{k}
{\delta}_i$ is the \emph{degree}, where ${\delta}_i =
{\max}_{1\leq j \leq n} \{ \deg g_{ij} \}$, $\mu = {\max}_{1\leq
i\leq k}\{{\delta}_i\}$ is the \emph{memory} and
$d_{f}=$wt$(C)=\min \{wt({\bf v}(D)) \mid {\bf v}(D) \in C, {\bf
v}(D)\neq 0 \}$ is the \emph{free distance} of the code.
\end{definition}

If ${\mathbb F}_{q}((D))$ is the field of Laurent series we define
the weight of ${\bf u}(D)$ as wt$({\bf u}(D)) =
{\sum}_{\mathbb{Z}}$wt$(u_i)$. A generator matrix $G(D)$ is called
\emph{catastrophic} if there exists a ${\bf u}{(D)}^{k}\in
{\mathbb F}_{q}{((D))}^{k}$ of infinite Hamming weight such that
${\bf u}{(D)}^{k}G(D)$ has finite Hamming weight. The
convolutional codes constructed in this paper have basic generator
matrices; consequently, they are non catastrophic.

We define the Euclidean inner product of two $n$-tuples ${\bf
u}(D) = {\sum}_i {\bf u}_i D^i$ and ${\bf v}(D) = {\sum}_j {\bf
u}_j D^j$ in ${\mathbb F}_q {[D]}^{n}$ as $\langle {\bf u}(D)\mid
{\bf v}(D)\rangle = {\sum}_i {\bf u}_i \cdot {\bf v}_i$. If $C$ is
a convolutional code then its (Euclidean) dual is given by
$C^{\perp }=\{ {\bf u}(D) \in {\mathbb F}_q {[D]}^{n}\mid \langle
{\bf u}(D)\mid {\bf v}(D)\rangle = 0$ for all ${\bf v}(D)\in C\}$.

\subsection{Convolutional Codes Derived From Block
Codes}\label{IIIA}

Let ${[n, k, d]}_{q}$ be a block code whose parity check matrix
$H$ is partitioned into $\mu +1$ disjoint submatrices $H_i$ such
that $H = \left[
\begin{array}{cccc}
H_0 & H_1 & \cdots & H_{\mu}\\
\end{array}\right]^{T},$ where each $H_i$ has $n$ columns, obtaining the
polynomial matrix
\begin{eqnarray}
G(D) =  {\tilde H}_0 + {\tilde H}_1 D + {\tilde H}_2 D^2 + \ldots
+ {\tilde H}_{\mu} D^{\mu}.
\end{eqnarray}
The matrix $G(D)$ generates a convolutional code $V$ with $\kappa$
rows, where $\kappa$ is the maximal number of rows among the
matrices $H_i$; the matrices ${\tilde H}_i$, where $0\leq i\leq
\mu$, are derived from the respective $H_i$ by adding zero-rows at
the bottom in such a way that the matrix ${\tilde H}_i$ has
$\kappa$ rows in total.

\begin{theorem}\cite[Theorem 3]{Aly:2007}\label{A}
Suppose that $C \subseteq {\mathbb F}_q^n$ is an ${[n, k, d]}_{q}$
code with parity check matrix $H \in {\mathbb F}_q^{(n-k)\times
n}$ partitioned into submatrices $H_0, H_1, \ldots, H_{\mu}$ as
above such that $\kappa =$ rk$H_0$ and rk$H_i \leq \kappa$ for $1
\leq i\leq \mu$. Then the matrix $G(D)$ is a reduced basic
generator matrix. Moreover, if $d_f$ and $d_f^{\perp}$ denote the
free distances of $V$ and $V^{\perp}$, respectively, $d_i$ denote
the minimum distance of the code $C_i = \{ {\bf v}\in {\mathbb
F}_q^n \mid {\bf v} {\tilde H}_i^t =0 \}$ and $d^{\perp}$ is the
minimum distance of $C^{\perp}$, then one has $\min \{ d_0 +
d_{\mu} , d \} \leq d_f^{\perp} \leq  d$ and $d_f \geq d^{\perp}$.
\end{theorem}


\section{The New Codes}\label{IV}

Constructions of convolutional codes have been appeared in the
literature
\cite{Piret:1988,Rosenthal:1999,Hole:2000,Rosenthal:2001,Gluesing:2006,Luerssen:2006,Schmale:2006,Luerssen:2008}.
It is not simple to derive families of such codes by means of
algebraic approaches. In other words, most of the convolutional
codes available in the literature are constructed case by case.

Motivated by the construction of new convolutional codes by means
of algebraic method, we propose the construction of new
convolutional codes derived from character codes. We reinforce
that the convolutional codes constructed in this paper have basic
generator matrices, so they are non catastrophic. Our first main
result is given in the following:

\begin{theorem}\label{main1}
Let ${\mathbb F}_{q}$ be a finite field of odd characteristic and
consider the commutative group $G=({\mathbb{Z}}_{2}^{m}, +)$ where
$m\geq 3$ is an integer. Assume that $r$ and $u$ are positive
integers such that the inequalities $ r < u < m$ and
$\displaystyle\sum_{i=u+1}^{m}\left(
\begin{array}{c}
m\\
i\\
\end{array}
\right) \geq \displaystyle\sum_{i=r+1}^{u}\left(
\begin{array}{c}
m\\
i\\
\end{array}
\right)$ hold. Then there exist unit memory convolutional codes
with parameters $(2^{m}, 2^{m}- s_{m}(u), s_{m}(u)- s_{m}(r); 1,
d_{f}\geq 2^{r+1} {)}_{q}$, where $s_{m}(u)$ and $s_{m}(r)$ are
given in Section~\ref{II}.
\end{theorem}

\begin{proof}
Assume that $n=2^{m}$, $r \geq 1$ is an integer, $X_{r}= \{ {\bf
x} \in {\mathbb{Z}}_{2}^{m} | wt_{H} ({\bf x}) > r \}$ and let
$C_{q}(r, m)=C_{X_{r}}$ be the character code with parameters
$[2^{m},$ $s_{m}(r), 2^{m-r}{]}_{q}$. Then a parity check matrix
of $C_{X_{r}}$ is given by
\begin{eqnarray*}
H_{X_{r}} =\\ \left[
\begin{array}{ccccc}
{\gamma}_{0}({\bf x}_{t_{m}}^{1}) & {\gamma}_{1}({\bf x}_{t_{m}}^{1}) &  \cdots & {\gamma}_{n-1}({\bf x}_{t_{m}}^{1}) \\
{\gamma}_{0}({\bf x}_{t_{m-1}}^{1}) & {\gamma}_{1}({\bf x}_{t_{m-1}}^{1}) &  \cdots & {\gamma}_{n-1}({\bf x}_{t_{m-1}}^{1}) \\
{\gamma}_{0}({\bf x}_{t_{m-1}}^{2}) & {\gamma}_{1}({\bf x}_{t_{m-1}}^{2}) &  \cdots & {\gamma}_{n-1}({\bf x}_{t_{m-1}}^{2}) \\
\vdots & \vdots & \vdots & \vdots\\
{\gamma}_{0}({\bf x}_{t_{m-1}}^{C_{m, m-1}}) & {\gamma}_{1}({\bf x}_{t_{m-1}}^{C_{m, m-1}}) &  \cdots & {\gamma}_{n-1}({\bf x}_{t_{m-1}}^{C_{m, m-1}}) \\
\vdots & \vdots & \vdots & \vdots\\
\vdots & \vdots & \vdots & \vdots\\
{\gamma}_{0}({\bf x}_{t_{r+1}}^{1}) & {\gamma}_{1}({\bf x}_{t_{r+1}}^{1})  & \cdots & {\gamma}_{n-1}({\bf x}_{t_{r+1}}^{1})  \\
{\gamma}_{0}({\bf x}_{t_{r+1}}^{2}) & {\gamma}_{1}({\bf x}_{t_{r+1}}^{2})  & \cdots & {\gamma}_{n-1}({\bf x}_{t_{r+1}}^{2})  \\
\vdots & \vdots & \vdots & \vdots\\
{\gamma}_{0}({\bf x}_{t_{r+1}}^{C_{m, r+1}}) & {\gamma}_{1}({\bf
x}_{t_{r+1}}^{C_{m, r+1}}) &
\cdots & {\gamma}_{n-1}({\bf x}_{t_{r+1}}^{C_{m, r+1}})\\
\end{array}
\right],
\end{eqnarray*}
where the elements ${\bf x}_{t_{r+j}}^{1}, {\bf x}_{t_{r+j}}^{2},
\ldots , {\bf x}_{t_{r+j}}^{C_{m, r+j}}$, $j=1, \ldots , m-r$, are
the $C_{m, r+j}=\left(
\begin{array}{c}
m\\
r+j\\
\end{array}
\right)$ elements in ${\mathbb{Z}}_{2}^{m}$ having Hamming weight
$r+j$.

For a positive integer $u$ with $r < u < m$, consider the set
$X_{u}= \{ {\bf x} \in {\mathbb{Z}}_{2}^{m} | wt_{H} ({\bf x})
> u \}$. Let $C_{q}(u, m)=C_{X_{u}}$ be the character code with
parameters $[2^{m},$ $s_{m}(u), 2^{m-u}{]}_{q}$. A parity check
matrix for $C_{X_{u}}$ is given by
\begin{eqnarray*}
H_{X_{u}} =\\ \left[
\begin{array}{ccccc}
{\gamma}_{0}({\bf x}_{t_{m}}^{1}) & {\gamma}_{1}({\bf x}_{t_{m}}^{1}) &  \cdots & {\gamma}_{n-1}({\bf x}_{t_{m}}^{1}) \\
{\gamma}_{0}({\bf x}_{t_{m-1}}^{1}) & {\gamma}_{1}({\bf x}_{t_{m-1}}^{1}) &  \cdots & {\gamma}_{n-1}({\bf x}_{t_{m-1}}^{1}) \\
{\gamma}_{0}({\bf x}_{t_{m-1}}^{2}) & {\gamma}_{1}({\bf x}_{t_{m-1}}^{2}) &  \cdots & {\gamma}_{n-1}({\bf x}_{t_{m-1}}^{2}) \\
\vdots & \vdots & \vdots & \vdots\\
{\gamma}_{0}({\bf x}_{t_{m-1}}^{C_{m, m-1}}) & {\gamma}_{1}({\bf x}_{t_{m-1}}^{C_{m, m-1}}) &  \cdots & {\gamma}_{n-1}({\bf x}_{t_{m-1}}^{C_{m, m-1}}) \\
\vdots & \vdots & \vdots & \vdots\\
\vdots & \vdots & \vdots & \vdots\\
{\gamma}_{0}({\bf x}_{t_{u+1}}^{1}) & {\gamma}_{1}({\bf x}_{t_{u+1}}^{1})  & \cdots & {\gamma}_{n-1}({\bf x}_{t_{u+1}}^{1})  \\
{\gamma}_{0}({\bf x}_{t_{u+1}}^{2}) & {\gamma}_{1}({\bf x}_{t_{u+1}}^{2})  & \cdots & {\gamma}_{n-1}({\bf x}_{t_{u+1}}^{2})  \\
\vdots & \vdots & \vdots & \vdots\\
{\gamma}_{0}({\bf x}_{t_{u+1}}^{C_{m, u+1}}) & {\gamma}_{1}({\bf
x}_{t_{u+1}}^{C_{m, u+1}}) &
\cdots & {\gamma}_{n-1}({\bf x}_{t_{u+1}}^{C_{m, u+1}})\\
\end{array}
\right],
\end{eqnarray*}
where the elements ${\bf x}_{t_{u+j}}^{1}, {\bf x}_{t_{u+j}}^{2},
\ldots , {\bf x}_{t_{u+j}}^{C_{m, u+j}}$, $j=1, \ldots , m-u$, are
the $C_{m, u+j}=\left(
\begin{array}{c}
m\\
u+j\\
\end{array}
\right)$ elements in ${\mathbb{Z}}_{2}^{m}$ having Hamming weight
$u+j$.

Since $u > r$, the parity check matrix $H_{X_{u}}$ is a submatrix
of $H_{X_{r}}$ and, consequently, we can split $H_{X_{r}}$ into
disjoint submatrices $H_{X_{u}}$ and $H$ as follows:

\begin{eqnarray*}
H_{X_{r}} =
\left[
\begin{array}{c}
H_{X_{u}}\\
--\\
H\\
\end{array}
\right] =\\ \left[
\begin{array}{ccccc}
{\gamma}_{0}({\bf x}_{t_{m}}^{1}) & {\gamma}_{1}({\bf x}_{t_{m}}^{1}) &  \cdots & {\gamma}_{n-1}({\bf x}_{t_{m}}^{1}) \\
{\gamma}_{0}({\bf x}_{t_{m-1}}^{1}) & {\gamma}_{1}({\bf x}_{t_{m-1}}^{1}) &  \cdots & {\gamma}_{n-1}({\bf x}_{t_{m-1}}^{1}) \\
{\gamma}_{0}({\bf x}_{t_{m-1}}^{2}) & {\gamma}_{1}({\bf x}_{t_{m-1}}^{2}) &  \cdots & {\gamma}_{n-1}({\bf x}_{t_{m-1}}^{2}) \\
\vdots & \vdots & \vdots & \vdots\\
{\gamma}_{0}({\bf x}_{t_{m-1}}^{C_{m, m-1}}) & {\gamma}_{1}({\bf x}_{t_{m-1}}^{C_{m, m-1}}) &  \cdots & {\gamma}_{n-1}({\bf x}_{t_{m-1}}^{C_{m, m-1}}) \\
\vdots & \vdots & \vdots & \vdots\\
\vdots & \vdots & \vdots & \vdots\\
{\gamma}_{0}({\bf x}_{t_{u+1}}^{1}) & {\gamma}_{1}({\bf x}_{t_{u+1}}^{1})  & \cdots & {\gamma}_{n-1}({\bf x}_{t_{u+1}}^{1})  \\
{\gamma}_{0}({\bf x}_{t_{u+1}}^{2}) & {\gamma}_{1}({\bf x}_{t_{u+1}}^{2})  & \cdots & {\gamma}_{n-1}({\bf x}_{t_{u+1}}^{2})  \\
\vdots & \vdots & \vdots & \vdots\\
{\gamma}_{0}({\bf x}_{t_{u+1}}^{C_{m, u+1}}) & {\gamma}_{1}({\bf
x}_{t_{u+1}}^{C_{m, u+1}}) &
\cdots & {\gamma}_{n-1}({\bf x}_{t_{u+1}}^{C_{m, u+1}})\\
---- & ---- & ---- & ----\\
{\gamma}_{0}({\bf x}_{t_{u}}^{1}) & {\gamma}_{1}({\bf x}_{t_{u}}^{1})  & \cdots & {\gamma}_{n-1}({\bf x}_{t_{u}}^{1})  \\
{\gamma}_{0}({\bf x}_{t_{u}}^{2}) & {\gamma}_{1}({\bf x}_{t_{u}}^{2})  & \cdots & {\gamma}_{n-1}({\bf x}_{t_{u}}^{2})  \\
\vdots & \vdots & \vdots & \vdots\\
{\gamma}_{0}({\bf x}_{t_{u}}^{C_{m, u}}) & {\gamma}_{1}({\bf
x}_{t_{u}}^{C_{m, u}}) & \cdots & {\gamma}_{n-1}({\bf x}_{t_{u}}^{C_{m, u}})\\
\vdots & \vdots & \vdots & \vdots\\
{\gamma}_{0}({\bf x}_{t_{r+1}}^{1}) & {\gamma}_{1}({\bf x}_{t_{r+1}}^{1})  & \cdots & {\gamma}_{n-1}({\bf x}_{t_{r+1}}^{1})  \\
{\gamma}_{0}({\bf x}_{t_{r+1}}^{2}) & {\gamma}_{1}({\bf x}_{t_{r+1}}^{2})  & \cdots & {\gamma}_{n-1}({\bf x}_{t_{r+1}}^{2})  \\
\vdots & \vdots & \vdots & \vdots\\
{\gamma}_{0}({\bf x}_{t_{r+1}}^{C_{m, r+1}}) & {\gamma}_{1}({\bf
x}_{t_{r+1}}^{C_{m, r+1}}) &
\cdots & {\gamma}_{n-1}({\bf x}_{t_{r+1}}^{C_{m, r+1}})\\
\end{array}
\right].
\end{eqnarray*}
Since $\displaystyle\sum_{i=u+1}^{m}\left(
\begin{array}{c}
m\\
i\\
\end{array}
\right) \geq \displaystyle\sum_{i=r+1}^{u}\left(
\begin{array}{c}
m\\
i\\
\end{array}
\right),$ it follows that rk$H_{X_{u}}\geq$ rk$H$. Then we form
the convolutional code $V$ generated by the reduced basic (see
Theorem~\ref{A}) generator matrix $$G(D)=\tilde H_{X_{u}}+ \tilde
H D,$$ where $\tilde H_{X_{u}} = H_{X_{u}}$ and $\tilde H$ is
obtained from $H$ by adding zero-rows at the bottom such that
$\tilde H$ has the number of rows of $H_{X_{u}}$ in total. By
construction, $V$ is a unit memory convolutional code. Since
$H_{X_{u}}$ is the parity check matrix of the code $C_{q}(u,
m)=C_{X_{u}}$, then $H_{X_{u}}$ has $2^{m}- s_{m}(u)$ linearly
independent rows, where $s_{m}(u)=
\displaystyle\sum_{i=0}^{u}\left(
\begin{array}{c}
m\\
i\\
\end{array}
\right)$, hence $V$ has dimension $k_{V}=2^{m}- s_{m}(u)$. From
construction, $H$ has $s_{m}(u)- s_{m}(r)$ linearly independent
rows and therefore $V$ has degree ${\delta}_{V}=s_{m}(u)-
s_{m}(r)$. From Theorem~\ref{A}, the free distance $d_{f}$ of $V$
satisfies $d_f \geq d^{\perp}$, where $d^{\perp}$ is the minimum
distance of the dual code ${[C_{q}(r, m)]}^{\perp}$. Since the
minimum distance of ${[C_{q}(r, m)]}^{\perp}$ is equal to
$2^{r+1}$, then $d_{f} \geq 2^{r+1}$. Therefore one can get an
$(2^{m}, 2^{m}- s_{m}(u), s_{m}(u)- s_{m}(r); 1, d_{f}\geq 2^{r+1}
{)}_{q}$ convolutional code.
\end{proof}

\begin{corollary}\label{cor1}
Let ${\mathbb F}_{q}$ be a finite field of odd characteristic and
consider the commutative group $G=({\mathbb{Z}}_{2}^{m}, +)$ where
$m\geq 3$ is an integer. Assume that $r$ and $u$ are positive
integers such that the inequalities $r < u < m$ and
$\displaystyle\sum_{i=u+1}^{m}\left(
\begin{array}{c}
m\\
i\\
\end{array}
\right) \geq \displaystyle\sum_{i=r+1}^{u}\left(
\begin{array}{c}
m\\
i\\
\end{array}
\right)$ hold. Then there exists an $(2^{m}, s_{m}(u), s_{m}(u)-
s_{m}(r) ; \mu, d_{f}\geq 2^{m-u}+1 {)}_{q}$ convolutional code
where, $s_{m}(u)$ and $s_{m}(r)$ are given above.
\end{corollary}

\begin{proof}
Assume the same notation utilized in the proof of
Theorem~\ref{main1}. We know that the convolutional code $V$
generated by the generator matrix $G(D)=\tilde H_{X_{u}}+ \tilde H
D$ has parameters $(2^{m}, 2^{m}- s_{m}(u), s_{m}(u)- s_{m}(r) ;
1, d_{f}\geq 2^{r+1} {)}_{q}$.

Consider the dual $V^{\perp}$ of the code $V$. We know that
$V^{\perp}$ has length $n=2^{m}$, dimension
$k_{V^{\perp}}=s_{m}(u)$ and degree $\delta= s_{m}(u)- s_{m}(r)$.
We need to compute the free distance $d_{f}^{\perp}$ of
$V^{\perp}$. From Theorem~\ref{A}, one has $\min \{ d_0 + d_1 , d
\} \leq d_{f}^{\perp} \leq  d$, where $d_0$ is the minimum
distance of the code with parity check matrix $H_{X_{u}}$, $d_1$
is the minimum distance of the code with parity check matrix $H$
and $d$ is the minimum distance of the code with parity check
matrix $H_{X_{r}}$. We know that $d_0 = 2^{m-u}$ and $d=2^{m-r}$.
Since the minimum distance $d_1$ is not known we have
$d_{f}^{\perp}\geq 2^{m-u}+1$.

Therefore there exists an $(2^{m}, s_{m}(u), s_{m}(u)- s_{m}(r);
\mu, d_{f}\geq 2^{m-u}+1 {)}_{q}$ convolutional code.
\end{proof}

\begin{example}
Consider that $m=5$, $u=2$ and $r=1$ and $q=3$. Thus the
inequality $\displaystyle\sum_{i=3}^{5}\left(
\begin{array}{c}
5\\
i\\
\end{array}
\right) = 16 > \displaystyle\sum_{i=2}^{2}\left(
\begin{array}{c}
5\\
i\\
\end{array}
\right)=10$ hold. From Theorem~\ref{main1}, there exists an $(32,
17, 10; 1, d_{f}\geq 4 {)}_{3}$ convolutional code. Moreover, from
Corollary~\ref{cor1}, there exists an $(32, 15, 10; \mu, d_{f}\geq
9 {)}_{3}$ convolutional code. On the other hand, if we take
$m=6$, $u=2$, $r=1$ and $q=3$, one can get $(64, 42, 15; 1,
d_{f}\geq 4 {)}_{3}$  $(64, 22, 15; \mu, d_{f}\geq 17 {)}_{3}$
convolutional codes.
\end{example}

Next we describe how to construct multi memory convolutional codes
derived from character codes.

\begin{theorem}\label{main2}
Let ${\mathbb F}_{q}$ be a finite field of odd characteristic and
consider the commutative group $G=({\mathbb{Z}}_{2}^{m}, +)$ where
$m\geq 4$ is an integer. Assume that $r, u, v$ are positive
integers such that the inequalities $r < v < u < m$,
$\displaystyle\sum_{i=u+1}^{m}\left(
\begin{array}{c}
m\\
i\\
\end{array}
\right) \geq \displaystyle\sum_{i=r+1}^{v}\left(
\begin{array}{c}
m\\
i\\
\end{array}
\right) \geq \displaystyle\sum_{i=v+1}^{u}\left(
\begin{array}{c}
m\\
i\\
\end{array}
\right)$ hold. Then there exist convolutional codes with
parameters $(2^{m}, 2^{m}- s_{m}(u), \delta ; 2, d_{f}\geq 2^{r+1}
{)}_{q}$, where $s_{m}(u)$ is given above and $\delta =
\displaystyle\sum_{i=r+1}^{v}\left(
\begin{array}{c}
m\\
i\\
\end{array}
\right)$.
\end{theorem}

\begin{proof}
Adopting the notation of Theorem~\ref{main1}, let us consider the
parity check matrices $H_{X_{r}}$ and $H_{X_{u}}$ of the codes
$C_{q}(r, m)$ and $C_{q}(u, m)$, respectively, such that
\begin{eqnarray*}
H_{X_{r}} = \left[
\begin{array}{c}
H_{X_{u}}\\
--\\
H_{1}\\
--\\
H_{2}\\
\end{array}
\right],
\end{eqnarray*}
where $H_{1}$ and $H_{2}$ are submatrices of $H_{X_{r}}$ with
$\displaystyle\sum_{i=v+1}^{u}\left(
\begin{array}{c}
m\\
i\\
\end{array}
\right)$ and $\displaystyle\sum_{i=r+1}^{v}\left(
\begin{array}{c}
m\\
i\\
\end{array}
\right)$ linearly independent rows, respectively. Note that this
is possible due to $r < v < u < m$. From hypothesis one has
rk$H_{X_{u}}\geq$ rk$H_{i}$, $i=1, 2$. Then we form the
convolutional code $V$ generated by the matrix
\begin{eqnarray*}
G(D)=\tilde H_{X_{u}}+ \tilde H_{1} D + \tilde H_{2} D^{2},
\end{eqnarray*}
where $\tilde H_{X_{u}} = H_{X_{u}}$. By construction, the code
$V$ is a two memory code of dimension $k_{V}=2^{m}- s_{m}(u)$.
Further, the degree of $V$ is equal to rk$\tilde
H_{2}=\displaystyle\sum_{i=r+1}^{v}\left(
\begin{array}{c}
m\\
i\\
\end{array}
\right)$. Moreover, from Theorem~\ref{A}, one has $d_f \geq
d^{\perp}$, that is, $d_{f} \geq 2^{r+1}$. Therefore one can get
an $(2^{m}, 2^{m}- s_{m}(u), \delta ; 2, d_{f}\geq 2^{r+1}
{)}_{q}$ convolutional code, where $\delta =
\displaystyle\sum_{i=r+1}^{v}\left(
\begin{array}{c}
m\\
i\\
\end{array}
\right)$.
\end{proof}

\begin{remark}
Note that Theorem~\ref{main1} can be straightforward generalized
in order to construct convolutional codes with memory $\mu \geq
3$. We do not present the generalization here since it is trivial.
\end{remark}

We now construct convolutional codes derived from group characters
codes $C_{q} (r, m; l)$ and their corresponding dual ${[C_{q} (r,
m; l)]}^{\perp}$. To proceed further we utilize the notation
$\left(
\begin{array}{c}
m\\
i\\
\end{array}
\right)_{l}$ to denote the cardinality of the set $X_{i} = \{ {\bf
x} \in {\mathbb{Z}}_{l}^{m}: ||{\bf x}||=i \}$, $0\leq i\leq
m(l-1)$, given by $\left(
\begin{array}{c}
m\\
i\\
\end{array}
\right)_{l}=\displaystyle\sum_{k=0}^{m}{(-1)}^{k}\left(
\begin{array}{c}
m\\
k\\
\end{array}
\right)$ $\left(
\begin{array}{c}
m-1+i-kl\\
m-1\\
\end{array}
\right)$. Now we are ready to show the next result:

\vspace{0.25cm}

\begin{theorem}\label{main3}
Consider the group $({\mathbb{Z}}_{l}^{m}, +)$, where $m\geq 3$,
$l\geq 3$ are integers, and let $F_{q}$ be a finite field such
that $l | (q - 1)$. Assume that $r$ and $u$ are positive integers
such that the inequalities $r < u < m(l-1)$ and
$\displaystyle\sum_{i=u+1}^{m}\left(
\begin{array}{c}
m\\
i\\
\end{array}
\right)_{l} \geq \displaystyle\sum_{i=r+1}^{u}\left(
\begin{array}{c}
m\\
i\\
\end{array}
\right)_{l}$ hold. Then there exists a unit memory convolutional
code with parameters $(l^{m}, l^{m}-S_{m}(u), S_{m}(u)- S_{m}(r);
1, d_{f}\geq (b+2)l^{a} {)}_{q}$, where $a$ and $b$ are integers
such that $r = a(l- 1) + b$, $0 \leq b \leq l-2$ and $S_{m}(u)$,
$S_{m}(u)$ are given in Section~\ref{II}.
\end{theorem}

\begin{proof} Assume that $X(r, m; l) = \{{\bf x} \in
{\mathbb{Z}}_{l}^{m} : ||{\bf x}|| > r \}$ and consider the
character code $C_{q} (r, m; l)$ with parameters ${[ l^{m},
S_{m}(r), (l - b)l^{m-1-a}]}_{q}$, where $0\leq r < m(l-1)$ is
writing as $r = a(l- 1) + b$, $0 \leq b \leq l-2$, and $S_{m}(r)=
\displaystyle\sum_{i=0}^{r}
\displaystyle\sum_{k=0}^{m}{(-1)}^{k}\left(
\begin{array}{c}
m\\
k\\
\end{array}
\right)\left(
\begin{array}{c}
m-1+i-kl\\
m-1\\
\end{array}
\right)$. A parity check matrix for $C_{q} (r, m; l)$ is given by
$H_{X(r, m; l)}=[{\gamma}_{j-1}({\bf x})]_{{\bf x}\in X(r, m; l),
1\leq j\leq l^{m}}.$ Next, consider the set $X(u, m; l) = \{{\bf
y} \in {\mathbb{Z}}_{l}^{m} : ||{\bf y}|| > u \}$, generating the
code $C_{q} (u, m; l)$ with parity check matrix $H_{X(u, m;
l)}=[{\gamma}_{j-1}({\bf y})]_{{\bf y}\in X(u, m; l), 1\leq j\leq
l^{m}}.$

Since $u > r$, the parity check matrix $H_{X(u, m; l)}$ is a
submatrix of $H_{X(r, m; l)}$ and, consequently, we can split the
latter matrix into disjoint submatrices $H_{X(u, m; l)}$ and $H$:
$$H_{X(r, m; l)}= \left[
\begin{array}{c}
H_{X(u, m; l)}\\
----\\
H\\
\end{array}
\right].$$ The matrices $H_{X(r, m; l)}$ and $H_{X(u, m; l)}$ have
rank $l^{m}-S_{m}(r)$ and $l^{m}-S_{m}(u)$, respectively, so $H$
has rank $S_{m}(u)- S_{m}(r)$. Because
$\displaystyle\sum_{i=u+1}^{m}\left(
\begin{array}{c}
m\\
i\\
\end{array}
\right)_{l} \geq \displaystyle\sum_{i=r+1}^{u}\left(
\begin{array}{c}
m\\
i\\
\end{array}
\right)_{l}$, one has rk$H_{X(u, m; l)}\geq$ rk$H$. Then one
obtains the convolutional code $V$ generated by
$$G(D)=\tilde H_{X(u, m; l)}+ \tilde H D.$$ From construction, $V$
is a unit memory convolutional code of length $l^{m}$ and
dimension $l^{m}-S_{m}(u)$. Additionally, $V$ has degree
$S_{m}(u)- S_{m}(r)$. We only need to compute (a lower bound to)
$d_{f}$. From Theorem~\ref{A}, $d_f \geq d^{\perp}$, where
$d^{\perp}$ is the minimum distance of ${[C_{q} (r, m;
l)]}^{\perp}$. Since the latter code is monomial equivalent to
$C_{q} (m(l-1)-1-r, m; l)$, it is easy to verify that
$d^{\perp}=(b+2)l^{a}$. Therefore, one obtains an $(l^{m},
l^{m}-S_{m}(u), S_{m}(u)- S_{m}(r); 1, d_{f}\geq (b+2)l^{a}
{)}_{q}$ convolutional code, as desired.
\end{proof}

\begin{remark}
Note that Theorem~\ref{main3} can be easily generalized in order
to construct multi-memory convolutional codes as well.
\end{remark}

\section{Code Comparisons}\label{V}

In this section we compare the parameters of the new convolutional
codes with the ones displayed in the literature. At the present,
it seems that the parameters of the (classical) convolutional
codes shown in \cite{LaGuardia:2012}, derived from BCH codes, are
the better ones. However, the referred constructions are valid
only to primitive BCH codes. Therefore, we compare the new code
parameters with the ones exhibited in \cite{Aly:2007}, since the
latter paper also brings good convolutional codes. We must note
that there exist other good convolutional codes in the literature,
but these codes are constructed case-by-case in most situations.

In Table~\ref{table1}, the new code parameters appear in the first
column and the parameters of the codes shown in \cite{Aly:2007}
are in the second column. Note that the new code parameters are
given by applying Theorem~\ref{main1} and Corollary~\ref{cor1}.

\begin{table}[!hpt]
\begin{center}
\caption{Code Comparisons \label{table1}}
\begin{tabular}{|c |c |}

\hline The new codes & Codes shown in \cite{Aly:2007}\\
\hline $(n, k, \gamma; \mu, d_{f}{)}_{q}$ & $(n, k^{*}, {\gamma}^{*}; 1, d_{f}^{*} {)}_{q}$\\
\hline
\hline ${(32, 15, 10; \mu, d_f\geq 9)}_{3}$ & ${(32, 16, \gamma; 1, d_f\geq 5)}_{3}$\\
\hline ${(64, 42, 15; 1, d_f\geq 4)}_{3}$ & ${(64, 32, \gamma; 1, d_f\geq 6)}_{3}$\\
\hline ${(64, 22, 15; \mu, d_f\geq 17)}_{3}$ & ${(64, 16, \gamma; 1, d_f\geq 8)}_{3}$\\
\hline ${(128, 64, 35; 1, d_f\geq 8)}_{3}$ & ${(128, 64, \gamma; 1, d_f\geq 6)}_{3}$\\
\hline ${(128, 64, 35; 1, d_f\geq 8)}_{3}$ & ${(128, 32, \gamma; 1, d_f\geq 8)}_{3}$\\
\hline ${(128, 64, 35; \mu, d_f\geq 17)}_{3}$ & ${(128, 32, \gamma; 1, d_f\geq 8)}_{3}$\\
\hline
\hline ${(32, 15, 10; \mu, d_f\geq 9)}_{5}$ & ${(32, 16, \gamma; 1, d_f\geq 5)}_{5}$\\
\hline ${(64, 42, 15; 1, d_f\geq 4)}_{5}$ & ${(64, 32, \gamma; 1, d_f\geq 5)}_{5}$\\
\hline ${(64, 22, 15; \mu, d_f\geq 17)}_{5}$ & ${(64, 16, \gamma; 1, d_f\geq 6)}_{5}$\\
\hline ${(128, 64, 35; 1, d_f\geq 8)}_{5}$ & ${(128, 64, \gamma; 1, d_f\geq 5)}_{5}$\\
\hline ${(128, 64, 35; 1, d_f\geq 8)}_{5}$ & ${(128, 32, \gamma; 1, d_f\geq 6)}_{5}$\\
\hline ${(128, 64, 35; \mu, d_f\geq 17)}_{5}$ & ----\\
\hline
\hline ${(32, 15, 10; \mu, d_f\geq 9)}_{7}$ & ${(32, 16, \gamma; 1, d_f\geq 8)}_{7}$\\
\hline ${(64, 42, 15; 1, d_f\geq 4)}_{7}$ & ${(64, 48, \gamma; 1, d_f\geq 5)}_{7}$\\
\hline ${(64, 22, 15; \mu, d_f\geq 17)}_{7}$ & ${(64, 8, \gamma; 1, d_f\geq 14)}_{7}$\\
\hline ${(128, 64, 35; \mu, d_f\geq 17)}_{7}$ & ${(128, 64, \gamma; 1, d_f\geq 8)}_{7}$\\
\hline ${(128, 64, 35; \mu, d_f\geq 17)}_{7}$ & ${(128, 16, \gamma; 1, d_f\geq 14)}_{7}$\\
\hline
\hline ${(32, 15, 10; \mu, d_f\geq 9)}_{9}$ & ${(32, 16, \gamma; 1, d_f\geq 8)}_{9}$\\
\hline ${(64, 42, 15; 1, d_f\geq 4)}_{9}$ & ${(64, 48, \gamma; 1, d_f\geq 5)}_{9}$\\
\hline ${(64, 22, 15; \mu, d_f\geq 17)}_{9}$ & ${(64, 8, \gamma; 1, d_f\geq 12)}_{9}$\\
\hline ${(128, 64, 35; \mu, d_f\geq 17)}_{7}$ & ${(128, 64, \gamma; 1, d_f\geq 8)}_{7}$\\
\hline ${(128, 64, 35; \mu, d_f\geq 17)}_{7}$ & ${(128, 16, \gamma; 1, d_f\geq 12)}_{7}$\\
\hline
\hline ${(32, 15, 10; \mu, d_f\geq 9)}_{11}$ & ${(32, 8, \gamma; 1, d_f\geq 6)}_{11}$\\
\hline ${(64, 42, 15; 1, d_f\geq 4)}_{11}$ & ${(64, 32, \gamma; 1, d_f\geq 5)}_{11}$\\
\hline ${(64, 22, 15; \mu, d_f\geq 17)}_{11}$ & ${(64, 16, \gamma; 1, d_f\geq 6)}_{11}$\\
\hline ${(128, 64, 35; 1, d_f\geq 8)}_{11}$ & ${(128, 64, \gamma; 1, d_f\geq 5)}_{11}$\\
\hline ${(128, 64, 35; \mu, d_f\geq 17)}_{11}$ & ${(128, 32, \gamma; 1, d_f\geq 6)}_{11}$\\
\hline
\end{tabular}
\end{center}
\end{table}

As can be seen in Table~\ref{table1}, the new codes have
parameters better than the ones shown in \cite{Aly:2007} for
almost all cases. Only in two cases the codes available in
\cite{Aly:2007} are better than the new codes.


\section{Summary}\label{VI}
We have constructed new families of convolutional codes derived
from group character codes. These codes are constructed
algebraically and not by computational search or even case by
case. Moreover, the new code parameters are better than the ones
available in the literature.

\section*{Acknowledgment}
This research has been partially supported by the Brazilian
Agencies CAPES and CNPq.

\end{document}